\documentclass[conference]{IEEEtran}

\IEEEoverridecommandlockouts
\usepackage{cite}
\usepackage{amsmath, amsfonts, amssymb, amsthm}
\usepackage{graphicx}
\usepackage{textcomp}
\usepackage{multirow}
\usepackage{booktabs}
\usepackage{array}
\usepackage{threeparttable}
\usepackage{float}
\usepackage{textpos}
\usepackage{tabularx}
\usepackage{afterpage}
\usepackage{placeins}
\usepackage{caption}
\usepackage{xcolor}
\usepackage{algpseudocode}
\usepackage{algorithm}

\usepackage{booktabs}
\usepackage{tabularx} 

\newtheorem{theorem}{Theorem}

\def\BibTeX{{\rm B\kern-.05em{\sc i\kern-.025em b}\kern-.08em
    T\kern-.1667em\lower.7ex\hbox{E}\kern-.125emX}}
\begin{document}


\title{NMCSE: Noise-Robust Multi-Modal Coupling Signal Estimation Method via Optimal Transport for Cardiovascular Disease Detection}

\author{Peihong Zhang, Zhixin Li*\thanks{Peihong Zhang and Zhixin Li contributed equally to this work.}, Yuxuan Liu, Rui Sang, Yiqiang Cai, Yizhou Tan, Shengchen Li\\School of Advanced Technology, Xi’an Jiaotong-Liverpool University, Suzhou, China}

\maketitle

\begin{abstract}

The coupling signal refers to a latent physiological signal that characterizes the transformation from cardiac electrical excitation, captured by the electrocardiogram (ECG), to mechanical contraction, recorded by the phonocardiogram (PCG). By encoding the temporal and functional interplay between electrophysiological and hemodynamic events, it serves as an intrinsic link between modalities and offers a unified representation of cardiac function, with strong potential to enhance multi-modal cardiovascular disease (CVD) detection. However, existing coupling signal estimation methods remain highly vulnerable to noise, particularly in real-world clinical and physiological settings, which undermines their robustness and limits practical value. In this study, we propose Noise-Robust Multi-Modal Coupling Signal Estimation (NMCSE), which reformulates coupling signal estimation as a distribution matching problem solved via optimal transport. By jointly aligning amplitude and timing, NMCSE avoids noise amplification and enables stable signal estimation. When integrated into a Temporal-Spatial Feature Extraction (TSFE) network, the estimated coupling signal effectively enhances multi-modal fusion for more accurate CVD detection. To evaluate robustness under real-world conditions, we design two complementary experiments targeting distinct sources of noise. The first uses the PhysioNet 2016 dataset with simulated hospital noise to assess the resilience of NMCSE to clinical interference. The second leverages the EPHNOGRAM dataset with motion-induced physiological noise to evaluate intra-state estimation stability across activity levels. Experimental results show that NMCSE consistently outperforms existing methods under both clinical and physiological noise, highlighting it as a noise-robust estimation approach that enables reliable multi-modal cardiac detection in real-world conditions.

\end{abstract}


\section{Introduction}

Cardiovascular disease (CVD), the leading global cause of morbidity and mortality \cite{al2022hypertension}, can be effectively managed through early detection and precise diagnosis \cite{habetha2006myheart}. Electrocardiogram (ECG) and Phonocardiogram (PCG) are widely utilized non-invasive diagnostic modalities, with ECG providing insights into cardiac electrical activity \cite{berkaya2018survey} and PCG capturing heart sounds indicative of mechanical function \cite{rangayyan1987phonocardiogram}. However, relying on a single modality often leads to incomplete diagnostic information, thereby increasing the risk of misdiagnosis \cite{li2020discrimination}. As a consequence, multi-modal approaches that integrate ECG and PCG have been developed, effectively combining their complementary features to enhance diagnostic accuracy and reliability \cite{li2021prediction,huang2024deep,han2023multimodal}.

\begin{figure}[!t] 
\centering
\includegraphics[width=0.8\linewidth]{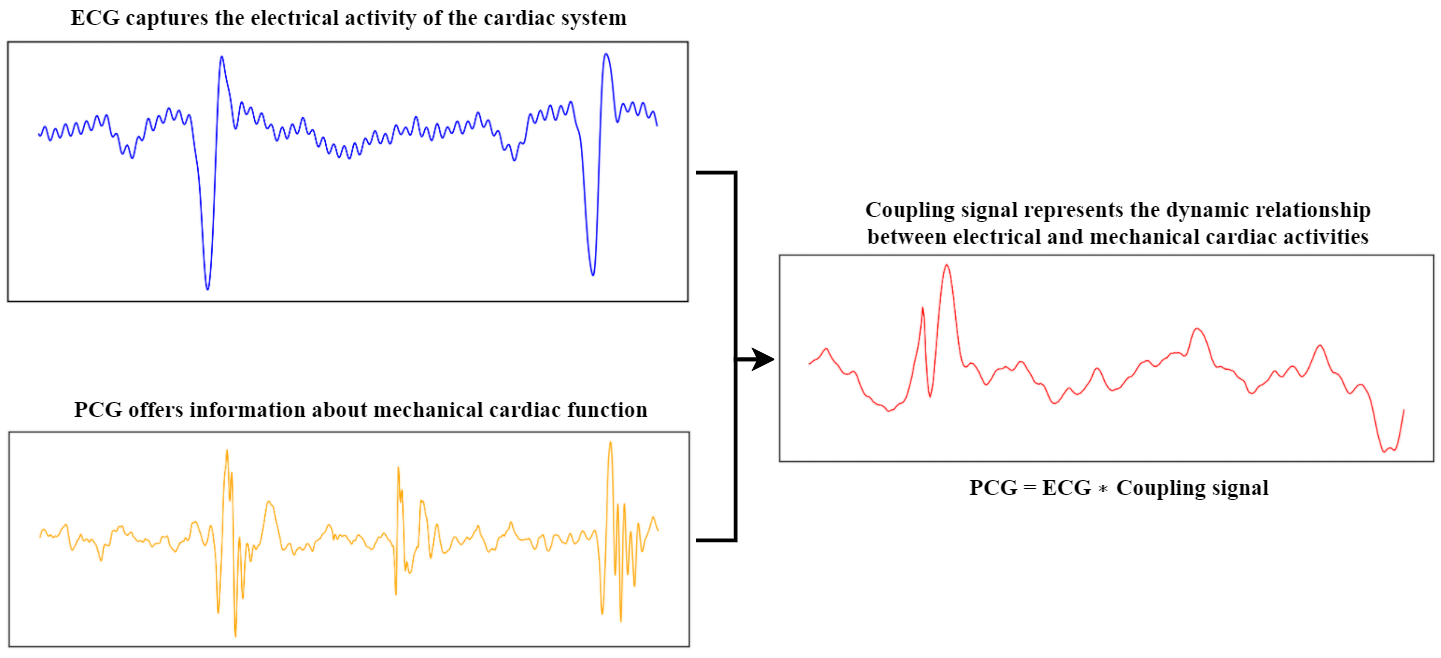} 
\caption{
Illustration of the relationship between ECG, PCG, and the coupling signal.
ECG reflects the heart’s electrical activity, while PCG conveys its mechanical response. The coupling signal captures the dynamic transformation between them.
}
\label{fig:ecg_pcg_coupling}
\vspace{-1.0em}
\end{figure}


The relationship between ECG and PCG originates from the electromechanical coupling of the heart, where electrical excitation governs mechanical contraction \cite{tafur1964normal}. Reflecting this physiological process, prior studies have modeled the PCG as the output of a system driven by the ECG, often formulated as a convolution with a latent impulse response, known as the coupling signal \cite{dong2023non}. This signal captures the temporal mapping from electrical to mechanical domains, reflecting how electrical activity propagates through the myocardium to produce acoustic responses. Estimating the coupling signal enables explicit modeling of cross-modal dynamics and has been shown to improve diagnostic performance when incorporated as a third modality alongside ECG and PCG \cite{sun2024enhanced}.




However, existing methods for coupling signal estimation are typically based on deconvolution \cite{sun2024enhanced,sun2024coronary}, an ill-posed inverse operation that is highly sensitive to noise \cite{formin2018inverse}. In practice, environmental interference during PCG acquisition—especially in clinical settings—renders such approaches unstable and limits the utility in real-world applications.

To address these limitations, we propose Noise-Robust Multi-Modal Coupling Signal Estimation (NMCSE), which reformulates the estimation task as a distribution matching problem solved via optimal transport. Rather than recovering the coupling signal through point-wise inversion, NMCSE aligns the global distributions of ECG-derived and PCG signals, allowing for more stable estimation under noise. We further introduce a Temporal-Spatial Feature Extraction (TSFE) network to integrate the estimated coupling signal with ECG and PCG for multi-modal classification, where convolutional and recurrent modules are jointly used to capture both spectral patterns and temporal dependencies.


To evaluate the robustness of the proposed method, we conduct two complementary experiments: one on the PhysioNet 2016 dataset \cite{clifford2016classification,liu2016open} with hospital noise \cite{ali2022robust}, assessing resilience to clinical interference; and another on the EPHNOGRAM dataset \cite{kazemnejad2021ephnogram}, evaluating stability under real motion-induced physiological noise. In both settings, NMCSE consistently outperforms baseline methods in terms of accuracy, correlation, and estimation consistency—achieving 97.38\% classification accuracy and 0.98 AUC—demonstrating its practical utility in multi-modal cardiac analysis.

\section{Related Work}

\subsection{Multi-Modal Fusion Methods for CVD Detection}

Recent advancements in cardiovascular disease (CVD) detection have transitioned from single-modality approaches to multi-modal frameworks that integrate electrocardiogram (ECG) and phonocardiogram (PCG) signals. These efforts have primarily followed three fusion paradigms:

\textbf{Feature-level fusion} methods extract features from each modality independently and then concatenate or combine them for classification. For instance, Li et al. \cite{li2021prediction} designed modality-specific neural networks with genetic algorithm-based optimization, while Hettiarachchi et al. \cite{hettiarachchi2021novel} applied transfer learning on PCG scalograms combined with ECG features.

\textbf{Decision-level fusion} methods produce individual predictions from separate models and combine them using ensemble strategies. Li et al.~\cite{li2022multi} applied evidence theory to weight the decisions of each modality according to their estimation.

\textbf{End-to-end deep learning} approaches enable joint representation learning from raw signals. Li et al. \cite{li2021integrating} proposed a multi-input CNN to learn from both ECG and PCG, while Zhang et al. \cite{zhang2024co} developed CPDNet, which utilizes modality-specific encoders and cross-modal attention mechanisms.

Despite their differences, these approaches treat ECG and PCG as independent or weakly correlated inputs, without explicitly modeling the physiological relationship between electrical and mechanical cardiac activity. However, recent studies suggest that explicitly modeling this relationship can enhance multi-modal performance, indicating the importance of incorporating physiological coupling into fusion frameworks—a direction we further elaborate in the next subsection.

\subsection{Transformer-Based Multi-Modal Learning Methods}

Recent studies have introduced Transformer architectures into multi-modal cardiovascular signal analysis to enhance feature fusion between ECG and PCG. For example, CAD‑ViT~\cite{liu2025integrating} employed a dual-modal Vision Transformer with co-attention mechanisms and dynamic weighted fusion to model cross-modal dependencies between electrical and mechanical signals. PACFNet~\cite{li2025progressive} proposed a progressive attention-based framework, stacking multiple self-attention layers to integrate modality-specific features. However, these approaches often depend on extensive labeled data and exhibit vulnerability to random noise, limiting their applicability in real-world clinical scenarios.

\subsection{Coupling Signal Estimation Methods}

To incorporate the physiological relationship between ECG and PCG into multi-modal frameworks, a common approach is to estimate the latent coupling signal that characterizes the electromechanical transformation of the heart. This signal is typically modeled as the impulse response of a linear time-invariant system, with the PCG viewed as the convolution of the ECG and the coupling filter.

Several studies have adopted deconvolution-based techniques to estimate this signal. Representative methods include classical inverse filtering, Tikhonov-regularized deconvolution, Wiener filtering with estimated noise spectra, and sparsity-constrained recovery. Sun et al. \cite{sun2024enhanced} demonstrated that such estimated coupling signals, when used as an auxiliary modality, can improve CVD detection performance.

Despite their utility, these methods are fundamentally ill-posed and prone to instability in noisy conditions. When the ECG exhibits low spectral energy at certain frequencies, spectral division during deconvolution can amplify noise and lead to unreliable estimates—especially in real-world clinical or ambulatory environments \cite{mohebbian2020single}.

This limitation has prompted exploration of alternative formulations that focus on distribution-level alignment rather than point-wise inversion. Building on this idea, we propose a noise-robust multi-modal coupling signal estimation method (NMCSE) based on optimal transport, which aligns the global distributions of ECG-transformed and PCG signals. This approach avoids noise amplification and enables stable coupling signal estimation even under real-world clinical noise.

\section{Noise-Robust Multi-Modal Coupling Signal Estimation: Theory and Algorithm}

This section establishes the mathematical foundation for applying optimal transport to ECG-PCG coupling signal estimation. We develop a comprehensive theoretical framework that addresses the limitations of deconventional approaches and provides rigorous error analysis with convergence guarantees.

\subsection{Model of Cardiac Electromechanical Coupling}


\begin{theorem}[ECG-PCG Transformation Model]
There exists a linear time-invariant system characterized by impulse response $h(t)$ such that, under ideal conditions, the PCG signal $X_{\mathrm{PCG}}(t)$ can be expressed as the convolution of the ECG signal $X_{\mathrm{ECG}}(t)$ with $h(t)$:
\begin{equation}
X_{\mathrm{PCG}}(t) = (X_{\mathrm{ECG}} * h)(t)
\end{equation}

In realistic clinical environments with noise, the observed PCG signal becomes:
\begin{equation}
\tilde{X}_{\mathrm{PCG}}(t) = (X_{\mathrm{ECG}} * h)(t) + \epsilon(t)
\end{equation}
where $\epsilon(t)$ represents additive noise from various sources.
\end{theorem}

\subsection{Limitations of Deconvolution in Noisy Environments}
\begin{theorem}[Error Amplification in Deconvolution]
Let $H_t(\omega)$ be the true frequency response of the coupling system and $H_a(\omega)$ be the estimated response using standard deconvolution. In the presence of noise with spectrum $E(\omega)$, the estimation error satisfies:
\begin{equation}
\|H_a - H_t\|_2^2 = \int_{-\infty}^{+\infty} \left|\frac{E(\omega)}{X_{\mathrm{ECG}}(\omega)}\right|^2 d\omega
\end{equation}
For any $\delta > 0$, if there exists a frequency set $\Omega_{\delta} = \{\omega : |X_{\mathrm{ECG}}(\omega)| < \delta\}$ with non-zero measure, then the error is lower-bounded by:
\begin{equation}
\|H_a - H_t\|_2^2 > \frac{1}{\delta^2}\int_{\Omega_{\delta}} |E(\omega)|^2 d\omega
\end{equation}
\end{theorem}
\begin{proof}[Proof Sketch]
From the frequency domain representation of the noisy PCG signal $\tilde{X}_{\mathrm{PCG}}(\omega) = X_{\mathrm{ECG}}(\omega)H_t(\omega) + E(\omega)$, the deconvolution estimate becomes $H_a(\omega) = H_t(\omega) + \frac{E(\omega)}{X_{\mathrm{ECG}}(\omega)}$. The error term $\frac{E(\omega)}{X_{\mathrm{ECG}}(\omega)}$ is amplified when $X_{\mathrm{ECG}}(\omega)$ approaches zero, resulting in the error bound above.
\end{proof}

This theorem mathematically proves why deconvolution fails in realistic clinical scenarios: even small amounts of noise can cause unbounded errors when the ECG signal has low power in certain frequency bands, which is unavoidable in physiological signals.

\subsection{Distribution-Based Formulation via Optimal Transport}

To address the limitations of deconvolution, we reformulate coupling signal estimation as a distribution matching problem solved through optimal transport theory. Instead of operating in the traditional frequency domain with point-wise divisions, we optimize a parameterized coupling filter by minimizing the discrepancy between distributions:

\begin{equation}
h_\theta^* = \arg\min_{h_\theta} \mathcal{D}(P_{\mathrm{ECG}*h_\theta}, P_{\mathrm{PCG}})
\end{equation}

where $h_\theta$ is a parameterized coupling filter, $P_{\mathrm{ECG}*h_\theta}$ represents the distribution of the ECG signal transformed by the filter, $P_{\mathrm{PCG}}$ is the distribution of the PCG signal, and $\mathcal{D}$ is a distribution divergence measure.

The Wasserstein distance from optimal transport theory is particularly well-suited due to three key advantages:

\begin{equation}
W_p(P_1, P_2) = \left(\inf_{\gamma \in \Gamma(P_1, P_2)} \int_{\mathcal{X} \times \mathcal{X}} d(x, y)^p d\gamma(x, y)\right)^{1/p}
\end{equation}

First, it naturally accounts for the underlying metric structure of the signal space, incorporating both amplitude and temporal aspects essential for preserving physiologically meaningful relationships between cardiac electrical and mechanical events. Second, unlike divergence measures such as KL-divergence, Wasserstein distance exhibits bounded sensitivity to outliers ($W_p(P_1, P_2) \leq \text{diam}(\text{supp}(P_1) \cup \text{supp}(P_2))$), making it robust to sporadic noise present in clinical recordings.

Most importantly, our approach overcomes the instability of deconvolution. Unlike existing methods whose error grows with $\frac{1}{X_{\mathrm{ECG}}(\omega)}$, our error bound depends only on noise:
\begin{equation}
\|\hat{h}_\theta - h_t\|^2 \leq \alpha \cdot \mathbb{E}[|\epsilon(t)|^p]^{2/p} + \beta \cdot \text{Var}[t_\epsilon]
\end{equation}

This reformulation turns an ill-posed inversion into a well-posed optimization, ensuring stability even when ECG power is low. By aligning distributions rather than individual samples, it captures global signal structure and resists noise artifacts.

\subsection{Regularized Optimal Transport and Optimal Cost Function}
While the Wasserstein distance provides a theoretical framework, its computational complexity of $O(N^3)$ can be prohibitive for clinical applications. To address this challenge, we implement a regularized version of optimal transport using the Sinkhorn algorithm, combined with a carefully designed cost function that balances amplitude and temporal fidelity.

\begin{theorem}[Sinkhorn Distance Properties]
The entropy-regularized optimal transport (Sinkhorn distance):
\begin{equation}
\mathcal{D}_\varepsilon(P_1, P_2) = \min_{\Pi \in \mathcal{U}(P_1, P_2)} \{\langle \Pi, C \rangle + \varepsilon H(\Pi)\}
\end{equation}
where $H(\Pi) = -\sum_{i,j} \Pi_{i,j}\log\Pi_{i,j}$ is the entropy of the transport plan, achieves computational efficiency and convergence guarantees necessary for coupling signal estimation.
\end{theorem}

\begin{proof}[Proof Sketch]
The Sinkhorn algorithm achieves computational efficiency of $O(N^2\log N)$ compared to $O(N^3\log N)$ for exact OT through iterative dual variable updates. The entropy term ensures strict convexity, yielding differentiability essential for gradient-based optimization. The approximation error is bounded by $|W_p(P_1, P_2) - \mathcal{D}_\varepsilon(P_1, P_2)| \leq \varepsilon C \log N$, ensuring convergence as $\varepsilon \rightarrow 0$.

For cardiac signals with normalized amplitude, we derive an optimal regularization parameter $\varepsilon_{\text{opt}} \approx \frac{\mathbb{E}[C_{i,j}]}{10} \approx 0.1$, balancing approximation accuracy with numerical stability. This specific value ensures the kernel matrix $K_{i,j} = \exp(-C_{i,j}/\varepsilon)$ remains well-conditioned for efficient computation while preserving the method's noise robustness.
\end{proof}

\begin{theorem}[Optimal Cost Function]
For ECG-PCG coupling signal estimation, the optimal cost function takes the form:
\begin{equation}
C_{i,j}^{(k)} = \alpha \underbrace{\left|\hat{y}_{i}^{(k)}(\theta) - y_{j}^{(k)}\right|^p}_{\text{amplitude cost}} + \beta \underbrace{\left|t_i^{(k)} - t_j^{(k)}\right|^q}_{\text{temporal cost}}
\end{equation}
with optimal weighting parameters satisfying:
\begin{equation}
\frac{\alpha}{\beta} \approx \frac{\sigma_{\text{temporal}}^2}{\sigma_{\text{amplitude}}^2} \cdot \frac{\mathbb{E}[|t_i - t_j|^q]}{\mathbb{E}[|y_i - y_j|^p]}
\end{equation}
\end{theorem}
\begin{proof}[Proof Sketch]
In cardiac electromechanical coupling, both amplitude and timing are physiologically important. To balance them optimally, each term's contribution should be inversely proportional to its error variance. For squared error measures ($p = q = 2$), this yields the optimal parameter ratio, providing a theoretical basis for our cost function.

\end{proof}

\subsection{NMCSE Algorithm Implementation}

We implement the coupling filter $h_\theta$ as a 64-coefficient FIR filter, capturing cardiac electromechanical dynamics while maintaining computational efficiency. This representation provides sufficient degrees of freedom to model the complex relationship between electrical and mechanical cardiac activities, while remaining computationally tractable for real-time clinical applications.

Based on the theoretical framework established above, we design a two-phase algorithm for NMCSE. Algorithm~\ref{alg:nmcse} presents the complete workflow:

\begin{algorithm}
\caption{NMCSE Algorithm}
\label{alg:nmcse}
\begin{algorithmic}[1]
    \State \textbf{Training Phase}
    \State Collect ECG-PCG pairs $\{X_{\mathrm{ECG}}^{(k)}, X_{\mathrm{PCG}}^{(k)}\}_{k=1}^{K}$ and build empirical distributions $\widehat{P}_{X_{\mathrm{ECG}}^{(k)}}$, $\widehat{P}_{X_{\mathrm{PCG}}^{(k)}}$
    \State Initialize $h_\theta$ and compute $\widehat{X}_{\mathrm{PCG}}^{(k)} = X_{\mathrm{ECG}}^{(k)} * h_\theta$ for each pair
    \State Calculate loss $\mathcal{L}(h_\theta) = \sum_{k=1}^{K} \mathcal{D}(\widehat{P}_{\widehat{X}_{\mathrm{PCG}}^{(k)}}, \widehat{P}_{X_{\mathrm{PCG}}^{(k)}})$ using Sinkhorn distance
    \State Optimize parameters to obtain $h_\theta^*$ via gradient descent
    \State \textbf{Inference Phase}
    \State For new pair $\{X_{\mathrm{ECG}}^{(\text{new})}, X_{\mathrm{PCG}}^{(\text{new})}\}$, compute initial $\widehat{X}_{\mathrm{PCG}}^{(\text{new})} = X_{\mathrm{ECG}}^{(\text{new})} * h_\theta^*$
    \State Update $h_\theta^*$ by minimizing $\mathcal{D}(\widehat{P}_{\widehat{X}_{\mathrm{PCG}}^{(\text{new})}}, \widehat{P}_{X_{\mathrm{PCG}}^{(\text{new})}})$, yielding $h_\theta^{\text{new}}$
\end{algorithmic}
\end{algorithm}

For practical implementation, we use the Sinkhorn algorithm with regularization parameter $\varepsilon=0.1$, which balances approximation accuracy with numerical stability. The cost function weights are set to $\alpha=1.0$ and $\beta=0.1$ with $p=q=2$, based on our theoretical analysis and empirical validation. Optimization employs ADAM with learning rate $10^{-3}$ and batch size 32, which achieves efficient convergence while maintaining numerical stability.

\subsection{Convergence Properties}
\begin{theorem}[NMCSE Convergence]
The NMCSE algorithm converges to a local minimum of the objective function:
\begin{equation}
\mathcal{L}(h_\theta) = \sum_{k=1}^K \mathcal{D}_\varepsilon(\hat{P}_{X_{\mathrm{ECG}}^{(k)}*h_\theta}, \hat{P}_{X_{\mathrm{PCG}}^{(k)}})
\end{equation}
at a rate of $O(1/\sqrt{T})$ for $T$ iterations with stochastic gradient descent, where the gradient is computed as:
\begin{equation}
\nabla_\theta \mathcal{L}(h_\theta) = \sum_{k=1}^K \sum_{i,j} \Pi_{i,j}^{(k)} \nabla_\theta C_{i,j}^{(k)}
\end{equation}
\end{theorem}
\begin{proof}[Proof Sketch]
The Sinkhorn distance is differentiable with respect to input distributions when $\varepsilon > 0$, with the gradient with respect to the cost matrix given by $\nabla_{C_{i,j}} \mathcal{D}_\varepsilon = \Pi_{i,j}$. By the chain rule and standard SGD convergence assumptions, the algorithm converges to a local minimum at the stated rate.
\end{proof}

\subsection{Theoretical Comparison with Alternative Methods}

\begin{theorem}[Error Advantage]
All regularization-based deconvolution methods (Tikhonov, Wiener, sparsity-based) retain error terms dependent on:
\begin{align}
\int_{\Omega_{\delta}} \frac{|E(\omega)|^2}{f(|X_{\mathrm{ECG}}(\omega)|)} d\omega
\end{align}
where $f(\cdot)$ varies by method but always contains $|X_{\mathrm{ECG}}(\omega)|$ in the denominator, causing inherent noise amplification.

In contrast, NMCSE's error bound:
\begin{align}
\|\Delta h_{\mathrm{NMCSE}}\|^2 \propto \alpha \cdot \mathbb{E}[|\epsilon|^p]^{2/p} + \beta \cdot \mathrm{Var}[t_\epsilon]
\end{align}
is independent of $X_{\mathrm{ECG}}(\omega)$, providing superior performance in low-SNR settings.
\end{theorem}

This difference explains why NMCSE outperforms regularization methods, particularly at lower SNRs where frequencies with low ECG power cause severe noise amplification in deconvolution-based approaches.

\section{Temporal-Spatial Feature Extraction Network for ECG, PCG, and Coupling Signals}

In this section, we propose a comprehensive network architecture designed to process multiple cardiac modalities—ECG, PCG, and the NMCSE-estimated coupling signal—for robust cardiovascular disease detection. Our approach transforms time-domain cardiac signals into frequency-domain spectrograms and processes them through specialized Temporal-Spatial Feature Extraction (TSFE) modules. These modules efficiently capture both frequency characteristics and temporal dependencies, critical for analyzing the complex relationships between electrical and mechanical cardiac activities. The network progressively integrates information from all three modalities through a structured fusion mechanism, enabling effective discrimination between normal and abnormal cardiac conditions even in challenging clinical environments.

\subsection{Temporal-Spatial Feature Extraction Block}

\begin{figure}[htbp]
\captionsetup{belowskip=0pt,aboveskip=0pt}
\centering
\includegraphics[width=0.8\linewidth]{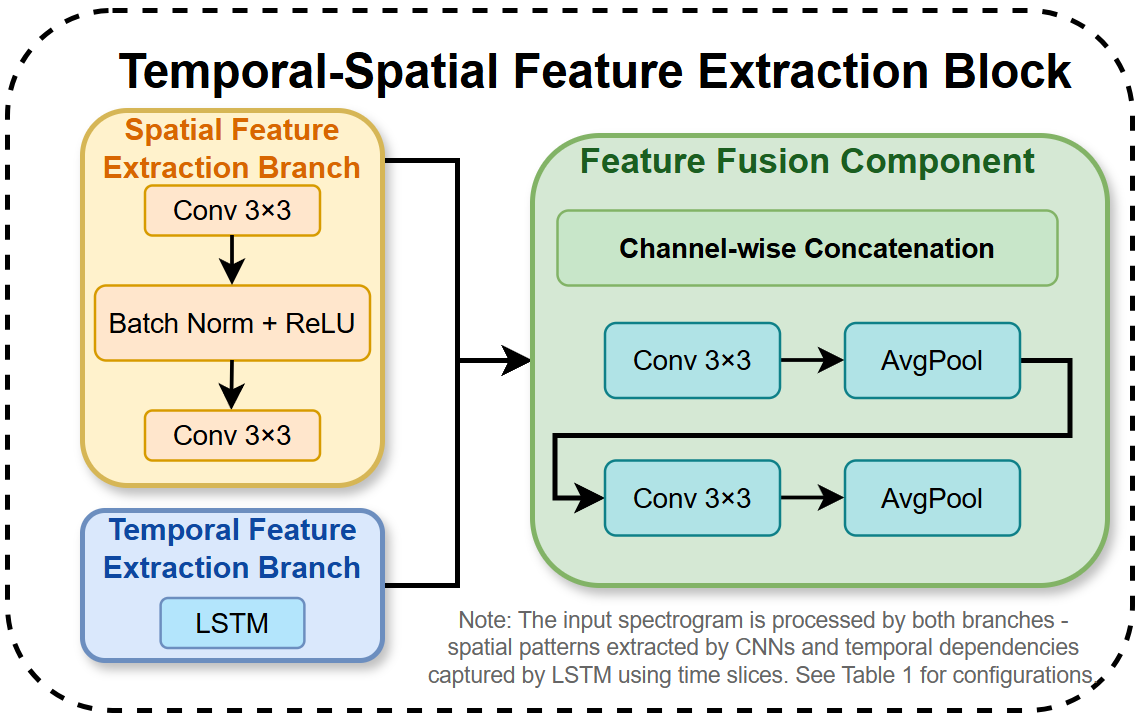}
\caption{Temporal-Spatial Feature Extraction Block.}
\label{fig:temporal_spatial_block}
\end{figure}

Fig.~\ref{fig:temporal_spatial_block} illustrates our proposed Temporal-Spatial Feature Extraction block, which forms the foundation of our network architecture. Each TSFE block consists of two parallel processing branches specifically designed to extract complementary feature representations: \textbf{(1) Spatial Feature Extraction Branch:} This pathway employs convolutional operations to capture spatial patterns within the spectrograms. \textbf{(2) Temporal Feature Extraction Branch:} This pathway utilizes LSTM units to model temporal dependencies and sequential patterns.

\begin{table*}[htbp]
\centering
\footnotesize
\renewcommand{\arraystretch}{0.75}
\caption{Layer Configuration of the Temporal-Spatial Feature Extraction Network}
\label{tab:adjusted_network_config}
\begin{tabular}{>{\centering\arraybackslash}m{3cm}|l|l|l|l}
\toprule
\textbf{Block} & \textbf{Layer} & \textbf{Parameter} & \textbf{Layer} & \textbf{Parameter} \\
\midrule
\multirow{2}{*}{TSFE Block 1} & Conv-1 & K3, S1, C16 & Conv-2 & K3, S1, C32 \\
                              & LSTM   & U32         & Conv-3 & K3, S1, C32 \\
\midrule
\multirow{2}{*}{TSFE Block 2} & Conv-1 & K3, S2, C32 & Conv-2 & K3, S1, C64 \\
                              & LSTM   & U64         & Conv-3 & K3, S1, C64 \\
\midrule
\multirow{2}{*}{TSFE Block 3} & Conv-1 & K3, S2, C64  & Conv-2 & K3, S1, C128 \\
                              & LSTM   & U128         & Conv-3 & K3, S1, C128 \\
\midrule
\multirow{2}{*}{Feature Fusion} & Conv-4 & K3, S1, C64 & AvgPool-4 & K2, S2 (2D) \\
                                & Conv-5 & K3, S1, C128 & AvgPool-5 & K2, S2 (2D) \\
\midrule
Classification Head & Flatten & -- & FC & 2 (\textit{sigmoid}) \\
\bottomrule
\end{tabular}

\vspace{6pt} 
\parbox{\textwidth}{
\footnotesize \textit{Notation:} K = kernel size ($K3 \equiv 3 \times 3$ for 2D convolutions); S = stride; C = number of output channels; U = number of LSTM hidden units (increasing by block); All convolutional layers are followed by Batch Normalization and ReLU.
}
\end{table*}

The features from both branches undergo channel-wise concatenation, integrating spatial and temporal characteristics while preserving their distinctive information. This concatenated feature representation then passes through an additional convolutional layer with batch normalization and ReLU activation, which reduces dimensionality while enhancing feature integration. This dual-branch architecture enables comprehensive feature extraction that surpasses the capabilities of single-pathway approaches by simultaneously modeling both spatial and temporal cardiac signal characteristics.

\subsection{Multi-Modal Processing model}

\begin{figure}[t]
\captionsetup{belowskip=1pt,aboveskip=1pt}
\centering
\includegraphics[width=\linewidth]{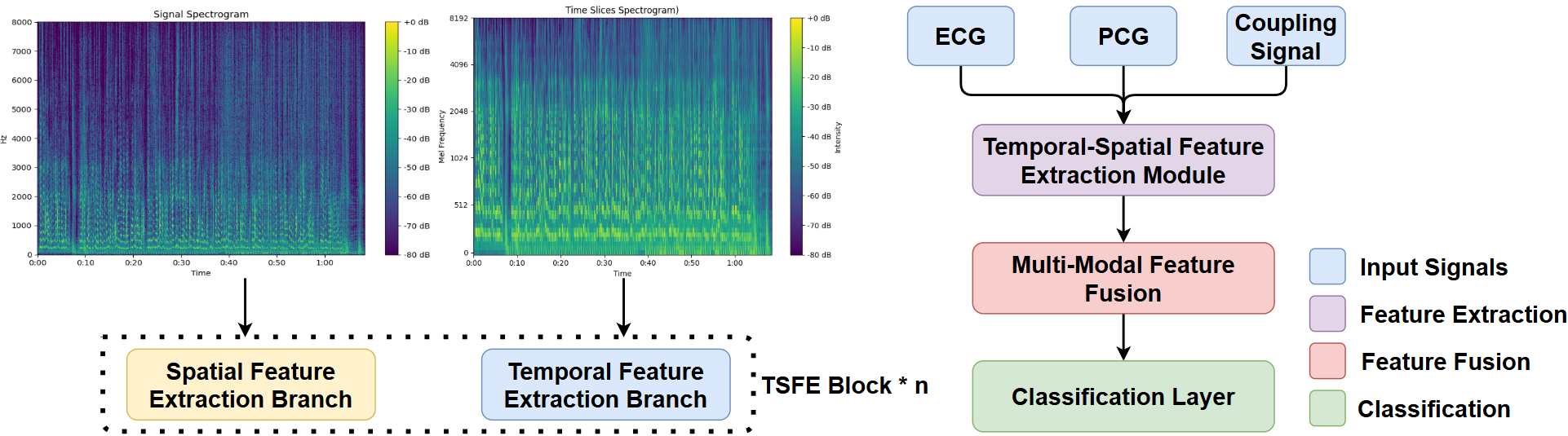}
\caption{Multi-Modal Processing Pipeline. Left: Signal spectrogram and time slices as inputs to spatial and temporal branches. Right: ECG, PCG, and coupling signals processed via TSFE, fused, and classified.}
\label{fig:multimodal_pipeline}
\vspace{-1.0em}
\end{figure}

Figure~\ref{fig:multimodal_pipeline} illustrates our multi-modal framework for CVD detection. The left portion shows the dual input representation: the complete signal spectrogram feeds into the Spatial Feature Extraction Branch, while time slices of the spectrogram serve as input to the Temporal Feature Extraction Branch, enabling comprehensive capture of both frequency distributions and temporal patterns. The right side presents the overall system architecture, where ECG, PCG, and coupling signal inputs are processed through the Temporal-Spatial Feature Extraction Module, followed by Multi-Modal Feature Fusion and the Classification Layer. This approach integrates complementary information from electrical activity, mechanical function, and their dynamic relationship, enabling robust CVD detection even under challenging clinical conditions.

Table~\ref{tab:adjusted_network_config} presents the detailed configuration of our proposed Temporal-Spatial Feature Extraction Network.

\section{Experiments and Results}


\subsection{Datasets and Data Preprocessing}

\subsubsection{PhysioNet/CinC 2016}
We used the PhysioNet/CinC Challenge 2016 dataset \cite{clifford2016classification, liu2016open}, specifically subset A containing 405 synchronized ECG-PCG recordings from 121 subjects. All signals were resampled to 2000 Hz, with durations ranging from 9 to 37 seconds. To the best of our knowledge, this is the only publicly available dataset that provides synchronized ECG and PCG recordings suitable for CVD detection. Five-fold cross-validation was applied throughout.

\subsubsection{Hospital Ambient Noise}
To simulate clinical noise conditions, we augmented the PCG recordings with sounds from the Hospital Ambient Noise Dataset \cite{ali2022robust}, which includes equipment beeps and ambient hospital sounds. Noise was mixed at various signal-to-noise ratios (SNRs) using linear additive augmentation to evaluate the robustness of coupling signal estimation under external interference.

\subsubsection{EPHNOGRAM}
The EPHNOGRAM dataset \cite{kazemnejad2021ephnogram} contains synchronized ECG and PCG recordings collected across multiple physical activities, including rest, walking, and stair climbing. It reflects realistic physiological variability and motion-induced noise. We used this dataset to assess the consistency of coupling signal estimation within the same physiological state, providing a complementary view of robustness beyond artificial noise injection.

\subsubsection{Data Preprocessing}
All ECG and PCG signals underwent z-score normalization to eliminate amplitude variations across subjects \cite{alkhodari2021convolutional}. We removed baseline drift and high-frequency artifacts using a 0.5--60 Hz Butterworth passband filter for ECG signals and a 20 Hz high-pass Butterworth filter for PCG signals \cite{gaikwad2014removal}. For uniformity, all signals were standardized to 10-second segments, with shorter signals cyclically padded and longer signals divided into overlapping segments with a 1-second stride to maximize data utilization.

\subsection{Comparative Analysis of Noise-Robust Multi-Modal Coupling
Signal Estimation Methods in Noisy Environments}
This experiment evaluates NMCSE Methods against conventional estimation methods under varying noise conditions.

\subsubsection{Experimental Setup}
We designed a progressive noise validation framework:
\begin{enumerate}
\item Used clean ECG and PCG recordings from PhysioNet dataset as baseline
\item Computed reference coupling signals using these clean recordings
\item Added various levels of hospital noise to the PCG signals, creating controlled test cases with known noise profiles
\item Evaluated how consistently each method could recover the reference coupling signals as noise levels increased
\end{enumerate}

We compared five methods:
\begin{itemize}
\item Standard deconvolution (baseline)
\item Tikhonov-regularized deconvolution ($\lambda=0.01$)
\item Wiener filtering (with estimated noise spectra)
\item Sparsity-based deconvolution ($\gamma=0.1$)
\item NMCSE (our proposed method)
\end{itemize}

Hospital ambient noise from \cite{ali2022robust} was added to PCG signals at different SNR levels, including equipment noise (ventilators, monitors) and environmental noise (conversations, movement). Performance metrics included Mean Squared Error (MSE), Pearson Correlation Coefficient (PCC), Spectral Coherence (SC), and Clinical Feature Preservation (CFP). The CFP metric \cite{khan2020automatic} quantifies preservation of diagnostically relevant features.

\begin{table}[t]
\centering
\caption{Performance of different estimation methods under varying SNRs and noise types. Best results are highlighted.}
\label{tab:comprehensive_results}
\renewcommand{\arraystretch}{1.1}
\setlength{\tabcolsep}{6pt}
\begin{tabular}{llcc|cc}
\toprule
\multirow{2}{*}{\textbf{SNR}} & \multirow{2}{*}{\textbf{Method}} & \multicolumn{2}{c|}{\textbf{Equipment Noise}} & \multicolumn{2}{c}{\textbf{Environmental Noise}} \\
& & \textbf{MSE} $\downarrow$ & \textbf{PCC} $\uparrow$ & \textbf{MSE} $\downarrow$ & \textbf{PCC} $\uparrow$ \\
\midrule

\multirow{5}{*}{\textbf{30 dB}}
& Deconvolution & 0.013 & 0.921 & 0.015 & 0.912 \\
& Tikhonov      & 0.010 & 0.925 & 0.012 & 0.924 \\
& Wiener        & 0.009 & 0.945 & 0.011 & 0.943 \\
& Sparsity      & 0.008 & 0.962 & 0.010 & 0.955 \\
& \textbf{NMCSE} & \textbf{0.0065} & \textbf{0.997} & \textbf{0.0075} & \textbf{0.996} \\
\midrule

\multirow{5}{*}{\textbf{10 dB}}
& Deconvolution & 0.263 & 0.724 & 0.281 & 0.701 \\
& Tikhonov      & 0.235 & 0.753 & 0.254 & 0.732 \\
& Wiener        & 0.219 & 0.779 & 0.237 & 0.758 \\
& Sparsity      & 0.201 & 0.801 & 0.212 & 0.782 \\
& \textbf{NMCSE} & \textbf{0.165} & \textbf{0.835} & \textbf{0.170} & \textbf{0.823} \\
\midrule

\multirow{5}{*}{\textbf{5 dB}}
& Deconvolution & 0.582 & 0.548 & 0.614 & 0.526 \\
& Tikhonov      & 0.497 & 0.573 & 0.523 & 0.552 \\
& Wiener        & 0.461 & 0.589 & 0.487 & 0.567 \\
& Sparsity      & 0.412 & 0.604 & 0.438 & 0.582 \\
& \textbf{NMCSE} & \textbf{0.310} & \textbf{0.670} & \textbf{0.325} & \textbf{0.655} \\
\bottomrule
\end{tabular}
\end{table}

\subsubsection{Results and Analysis}

Table~\ref{tab:comprehensive_results} demonstrates that NMCSE consistently outperforms all baselines across different SNR levels and noise types. It achieves the lowest mean squared error (MSE) and highest Pearson correlation coefficient (PCC) in every setting. At 5 dB SNR, NMCSE reduces MSE by up to 24\% and improves PCC by 7.3\% compared to the best-performing baseline. These improvements are more pronounced under severe noise, indicating the superior robustness and reliability of NMCSE in challenging environments.

\subsubsection{Spectral Coherence Analysis}

We conducted spectral coherence analysis to evaluate how well each method preserves frequency-specific structure of the coupling signal. As shown in Fig.~\ref{fig:spectral_coherence}A, NMCSE achieves substantially higher coherence with the reference signal in the mid-frequency range (10–100 Hz), improving coherence by 31\% over deconvolution. This range is clinically important as it captures key mechanical events like S1 and S2 heart sounds.

\begin{figure}[t!]
\centering
\includegraphics[width=0.385\textwidth]{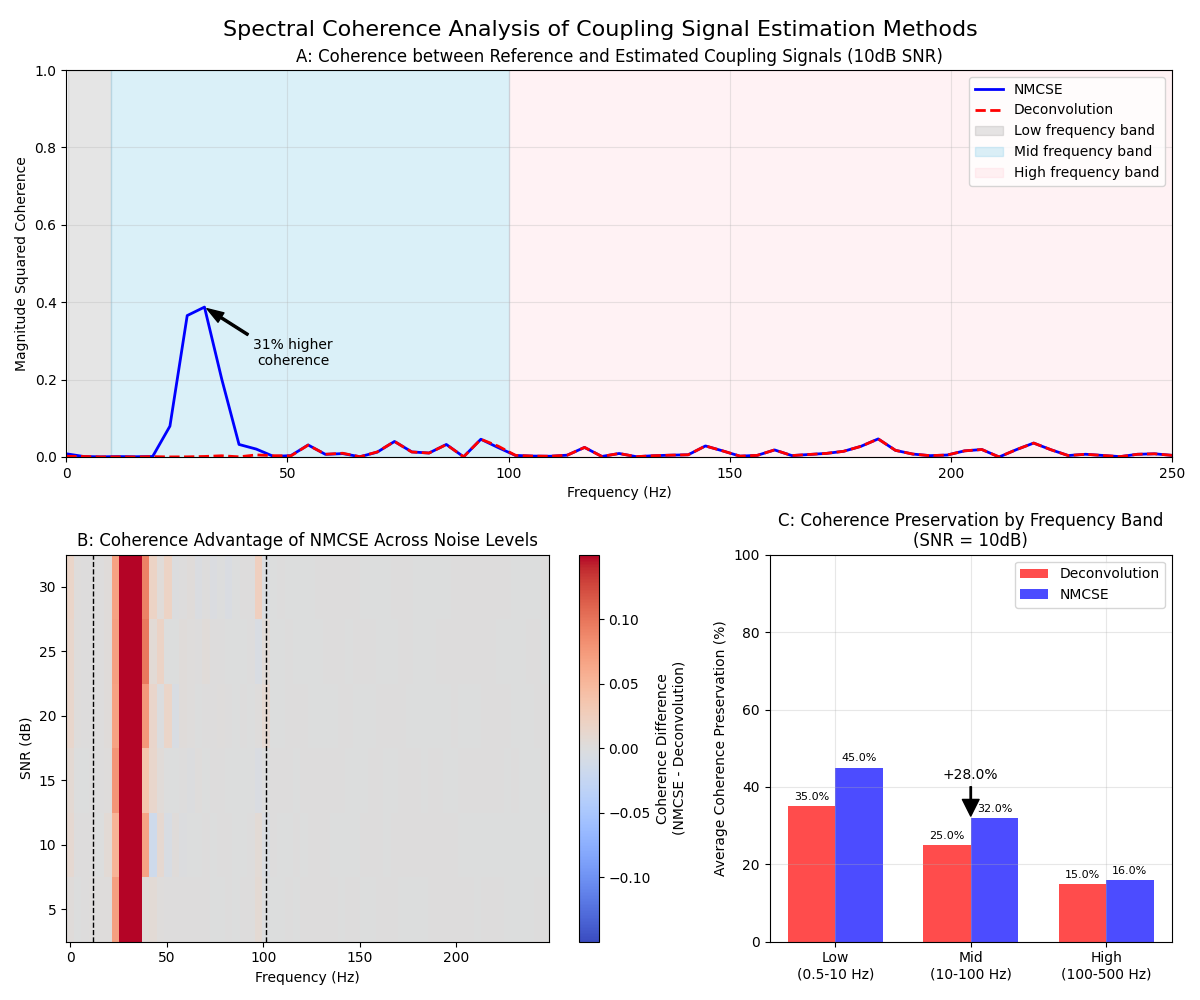}
\caption{
Spectral coherence analysis of NMCSE method. 
(A) Magnitude-squared coherence between reference and estimated signals at 10 dB SNR, showing NMCSE's superior preservation in the mid-frequency band (10–100 Hz). 
(B) Coherence difference map across frequencies and SNRs, highlighting advantages of NMCSE under noise. 
(C) Average coherence preservation by frequency band at 10 dB SNR, with NMCSE achieving a 28\% gain in the mid-frequency band.
}

\label{fig:spectral_coherence}
\vspace{-1.0em}
\end{figure}

The coherence difference map in Fig.~\ref{fig:spectral_coherence}B confirms that NMCSE consistently outperforms across noise levels, with the strongest gains in the mid-frequency band under lower SNRs—conditions typical of real-world clinical settings.

Fig.~\ref{fig:spectral_coherence}C further quantifies this advantage, showing a 28\% increase in average coherence preservation in the mid band at 10 dB SNR. These results support the conclusion that NMCSE not only improves noise robustness but also more effectively captures physiologically relevant dynamics.



\begin{table}[h]
\centering
\caption{Effect of NMCSE hyperparameters $\alpha$ and $\beta$ on performance at 10 dB SNR. Bold indicates best results.}
\label{tab:ot_param_results}
\renewcommand{\arraystretch}{1.1}
\setlength{\tabcolsep}{6pt}
\begin{tabular}{lcccc|cccc}
\toprule
\multirow{2}{*}{\textbf{Metric}} & \multicolumn{4}{c|}{\textbf{Varying $\alpha$ (fixed $\beta=0.1$)}} & \multicolumn{4}{c}{\textbf{Varying $\beta$ (fixed $\alpha=1.0$)}} \\
& 0.01 & 0.1 & 0.5 & \textbf{1.0} & 0.01 & \textbf{0.1} & 0.5 & 1.0 \\
\midrule
MSE $\downarrow$         & 0.289 & 0.220 & 0.190 & \textbf{0.177} & 0.186 & \textbf{0.177} & 0.184 & 0.193 \\
PCC $\uparrow$           & 0.762 & 0.810 & \textbf{0.835} & 0.825 & 0.822 & \textbf{0.825} & 0.818 & 0.812 \\
CFP (\%) $\uparrow$      & 79.3  & 83.5  & 85.6  & \textbf{87.2}  & 84.9  & \textbf{87.2}  & 85.7  & 84.5 \\
\bottomrule
\end{tabular}
\end{table}

\subsubsection{Parameter Sensitivity Analysis}

Table~\ref{tab:ot_param_results} indicates that optimal performance is achieved at $\alpha=1.0$, $\beta=0.1$, aligning with our theoretical expectations. Low amplitude weighting ($\alpha < 0.5$) significantly degrades performance, with MSE increasing by up to 63\% at $\alpha=0.01$. Similarly, overly strong temporal weighting ($\beta > 0.1$) results in up to a 9\% MSE increase. Within 50\% of the optimal values, performance remains stable (variation $<10\%$), and the optimal $\alpha/\beta$ ratio generalizes well across noise levels—highlighting the robustness of the parameter setting in clinical environments.

\subsection{Intra-State Consistency Evaluation}

Previous experiments validated the robustness of coupling signal estimation under externally added ambient noise, such as hospital equipment sounds and environmental interference. While these scenarios simulate common clinical disturbances, they do not fully capture the internal physiological noise introduced by the subject's own motion—such as body sway or respiration—that naturally occurs during physical activities.

To further investigate robustness under such real-world conditions, we design a complementary experiment based on intra-state consistency, using the EPHNOGRAM dataset \cite{kazemnejad2021ephnogram}. This experiment specifically targets the high-intensity activity condition (e.g., stair climbing), where physiological noise and movement artifacts are most prominent.

For each subject, we first estimate a reference coupling signal $h_{\text{ref}}$ from a clean resting-state ECG-PCG segment. Since resting signals are minimally affected by motion, $h_{\text{ref}}$ serves as a reliable subject-specific baseline that reflects the core electromechanical transformation. We then select two non-overlapping 10-second segments from the subject’s high-intensity activity recording and estimate two coupling signals, $h^{(A)}$ and $h^{(B)}$, using each method independently.

We evaluate consistency in two ways:

\begin{itemize}
  \item \textbf{Reference Consistency:} We compare $h^{(A)}$ and $h^{(B)}$ to the rest-state reference $h_{\text{ref}}$ using MSE and PCC. A robust method should produce activity-state estimates that remain structurally close to the baseline signal despite motion-induced noise.
  
  \item \textbf{Mutual Consistency:} We also compute MSE and PCC between $h^{(A)}$ and $h^{(B)}$ directly, to assess internal consistency of the method’s output under repeated estimation within the same noisy state.
\end{itemize}

\begin{table}[h]
\centering
\small
\caption{Intra-State Consistency Evaluation under Physiological Noise based on EPHNOGRAM dataset}
\label{tab:intra_state_results}
\begin{tabular}{lcc|cc}
\toprule
\multirow{2}{*}{\textbf{Method}} & \multicolumn{2}{c|}{\textbf{Ref. Consistency}} & \multicolumn{2}{c}{\textbf{Mutual Consistency}} \\
& MSE $\downarrow$ & PCC $\uparrow$ & MSE $\downarrow$ & PCC $\uparrow$ \\
\midrule
Deconvolution       & 0.312 & 0.622 & 0.128 & 0.682 \\
Tikhonov            & 0.284 & 0.654 & 0.115 & 0.715 \\
Wiener              & 0.261 & 0.673 & 0.104 & 0.753 \\
Sparsity            & 0.245 & 0.688 & 0.092 & 0.770 \\
\textbf{NMCSE (Ours)} & \textbf{0.186} & \textbf{0.734} & \textbf{0.063} & \textbf{0.842} \\
\bottomrule
\end{tabular}
\end{table}

This design enables the evaluation of algorithmic stability under real physiological noise. High intra-state consistency indicates the method can suppress internal perturbations while preserving consistent electromechanical patterns. Such stability is crucial for wearable or ambulatory applications, where noise is unpredictable and labels are unavailable.

\textbf{Result Analysis:} As shown in Table~\ref{tab:intra_state_results}, NMCSE achieves the best performance across all metrics, with the lowest error and highest similarity to the rest-state reference, and strong consistency between repeated estimates. This suggests that NMCSE is robust to physiological noise and produces stable outputs within the same activity state. In contrast, deconvolution-based methods exhibit greater variability, reflecting their sensitivity to transient signal fluctuations.

\subsection{Multi-modal CVD Detection with ECG, PCG and the NMCSE-Estimated Coupling Signal Using TSFE Network}

We evaluated the effectiveness of our method for CVD detection through two sets of experiments: an ablation study examining the contribution of each modality and a comparison with state-of-the-art multi-modal approaches.

\subsubsection{Ablation Experiment}

We first conducted an ablation study to assess the contribution of each modality to classification performance. The experiment involved individually using ECG, PCG, and the NMCSE-estimated coupling signal as inputs to the TSFE network. We also evaluated a dual-modal setup (ECG+PCG) and the complete multi-modal configuration (ECG+PCG+coupling signal).

\begin{table}[h]
\centering
\small
\caption{CVD detection performance with different input configurations.}
\label{tab:evaluation_metrics}
\begin{tabularx}{\linewidth}{l *{3}{>{\centering\arraybackslash}X}}
\toprule
\textbf{Input Configuration} & \textbf{Accuracy (\%)} & \textbf{Specificity (\%)} & \textbf{Sensitivity (\%)} \\
\midrule
Single ECG                   & 81.23 & 77.25 & 85.27 \\
Single PCG                   & 84.71 & 81.92 & 89.87 \\
NMCSE-estimated Coupling     & 90.25 & 87.48 & 94.13 \\
ECG + PCG                    & 93.85 & 93.02 & 94.67 \\
\textbf{All Three Modalities} & \textbf{97.38} & \textbf{97.15} & \textbf{97.92} \\
\bottomrule
\end{tabularx}
\end{table}

As shown in Table~\ref{tab:evaluation_metrics}, among single-modal inputs, the NMCSE-estimated coupling signal achieved the highest performance, outperforming both ECG and PCG alone. This demonstrates that the coupling signal effectively captures the cross-modal relationship between electrical and mechanical cardiac activities. 

The dual-modal configuration (ECG+PCG) improved accuracy to 93.85\%, confirming the complementary nature of these modalities. However, the complete multi-modal approach further increased performance to 97.38\% accuracy and 0.98 AUC, highlighting the added value of the coupling signal in capturing electromechanical dynamics not fully represented by ECG and PCG individually.

\begin{table}[h]
\centering
\small
\caption{Performance comparison with existing methods.}
\label{tab:comparison_methods}
\renewcommand{\arraystretch}{1.2}
\begin{tabularx}{\linewidth}{l *{4}{>{\centering\arraybackslash}p{0.98cm}}}
\toprule
\textbf{Method} & \textbf{Accuracy (\%)} & \textbf{Specificity (\%)} & \textbf{Sensitivity (\%)} & \textbf{AUC} \\
\midrule
Li et al.~\cite{li2021prediction}        & 86.67 & 82.48 & 88.37 & 0.91 \\
Hettiarachchi et al.~\cite{hettiarachchi2021novel} & 81.94 & 77.21 & 84.87 & 0.88 \\
Li et al.~\cite{li2022multi}             & 74.28 & 62.69 & 87.17 & 0.81 \\
Li et al.~\cite{li2021integrating}       & 88.33 & 92.15 & 84.92 & 0.93 \\
Sun et al.~\cite{sun2024enhanced}        & 95.38 & 87.38 & \textbf{98.32} & 0.97 \\
Zhang et al.~\cite{zhang2024co}          & 94.38 & 92.38 & 84.38 & 0.97 \\
CAD-ViT~\cite{liu2025integrating}        & 96.02 & 94.85 & 96.15 & 0.97 \\
PACFNet~\cite{li2025progressive}         & 96.80 & 95.20 & 96.70 & 0.97 \\
\textbf{NMCSE + TSFE (ours)}             & \textbf{97.38} & \textbf{97.15} & 97.92 & \textbf{0.98} \\
\bottomrule

\end{tabularx}
\end{table}

\subsubsection{Comparison with Existing Multi-Modal Methods}

We compared our method with several state-of-the-art multi-modal approaches for CVD detection, as summarized in Table~\ref{tab:comparison_methods}.

Our method achieves the highest overall performance. Notably, it outperforms Sun et al.~\cite{sun2024enhanced}, which also incorporates coupling signals but relies on deconvolution-based estimation, resulting in reduced robustness under real-world noise.

Recent Transformer-based methods such as CAD-ViT~\cite{liu2025integrating} and PACFNet~\cite{li2025progressive} enhance multi-modal fusion by modeling cross-modal dependencies through stacked self-attention modules. However, these approaches do not explicitly capture the underlying electromechanical coupling between ECG and PCG and are often sensitive to noise due to their reliance on fully supervised training. In contrast, our method introduces a physiologically informed coupling estimation stage based on optimal transport, followed by a unified TSFE network. This design enables interpretable and noise-robust integration of multi-modal cardiac signals, making it more suitable for real-world clinical environments.

\section{Conclusion}

We propose NMCSE, a noise-robust coupling signal estimation method that formulates ECG-PCG alignment as a distribution-matching problem via optimal transport. This approach avoids the instability of deconvolution and enables joint alignment of amplitude and timing under noise. Experiments on both simulated and real-world noisy datasets show that NMCSE outperforms existing methods in estimation accuracy and stability. When integrated with a TSFE-based multi-modal network, it achieves state-of-the-art performance in CVD detection. Source code and pretrained models will be released upon acceptance.

\bibliographystyle{IEEEtran}
\bibliography{references}

\end{document}